\documentclass[11pt,a4paper,english]{article}

\usepackage{amsthm}
\usepackage{amsmath}
\usepackage{algorithm}
\usepackage[noend]{algpseudocode}
\usepackage[hidelinks]{hyperref}
\usepackage{svg}
\usepackage{calc}
\usepackage{array}

\newtheorem{theorem}{Theorem}
\newtheorem{definition}{Definition}
\newtheorem{remark}{Remark}

\newcommand{\cdatree}{\mathsf{cda\_tree}}
 \newcommand{\cdalist}{\mathsf{cda\_list}}
\newcommand{\tsp}{\mathsf{timestamp}}
 \newcommand{\pr}{\mathsf{price}}
 
 \newcommand{\q}{\mathsf{qty}}
 \newcommand{\id}{\mathsf{id}}

 \newcommand{\ids}{\mathsf{ids}}

 \newcommand{\iter}{\mathsf{Iterated}}

\title{Efficient and Verified Continuous Double Auctions}

\author{Mohit Garg\thanks{Indian Institute of Science, Bengaluru, India. mohitgarg@iisc.ac.in. Supported by a fellowship from the Walmart Center for Tech Excellence at IISc (CSR Grant WMGT-23-0001) and the SERB Core Research Grant (CRG/2022/001176) on `Optimization under Intractability and Uncertainty'. A part of this work was done while the author was affiliated with the University of Bremen, Germany.} \ and
Suneel Sarswat\thanks{Tata Institute of Fundamental Research, Mumbai, India. suneel.sarswat@gmail.com.}}

\date{}

\begin{document}

\maketitle

\begin{abstract}Continuous double auctions are commonly used to match orders at currency, stock, and commodities exchanges. A verified implementation of continuous double auctions is a useful tool for market regulators as they give rise to automated checkers that are guaranteed to detect errors in the trade logs of an existing exchange if they contain trades that violate the matching rules. We provide an efficient and formally verified implementation of continuous double auctions that takes $O(n \log n)$ time to match $n$ orders. This improves an earlier $O(n^2)$ verified implementation. We also prove a matching $\Omega(n\log n)$ lower bound on the running time for continuous double auctions. Our new implementation takes only a couple of minutes to run on ten million randomly generated orders as opposed to a few days taken by the earlier implementation. Our new implementation gives rise to an efficient automatic checker. 

We use the Coq proof assistant for verifying our implementation and extracting a verified OCaml program. While using Coq's standard library implementation of red-black trees to obtain our improvement, we observed that its specification has serious gaps, which we fill in this work; this might be of independent interest.
\end{abstract}

\section{Introduction}
Continuous double auctions are used to match buy and sell orders for a particular product at an exchange. For example, they are used at currency, stock, and commodities exchanges. These exchanges generally use computer software to match orders and are required to adhere to regulatory directives that ensure fairness, safety, and transparency. There are multiple reported incidents where exchanges were found violating the regulatory directives or the stated rules~\cite{nyse1,nyse2,ubs,nse}. Bugs in the exchange program can trigger undesirable events leading to significant losses. 

These problems have received attention from the formalization community and form an active research area~\cite{PI17,SS20,RSS21,P21,GS22,clf}. In~\cite{GS22}, a general model of continuous double auctions is considered, and three simple properties of continuous double auctions are identified and shown to be sufficient to uniquely determine the input-output relation, thus yielding formal specifications for such auctions. A natural algorithm for continuous double auctions is then verified in the Coq theorem prover. As a novel application, the verified program obtained is used to build an automated checker that goes over the trade logs of an exchange, comparing the matchings generated by the exchange against the matchings produced by the verified program, and reporting any mismatches. Given that the specifications ensure a unique input-output relation, a mismatch would mean that the exchange program violates the specifications. Such verified programs and accompanying checkers can be extremely useful to the exchanges and the regulators.

\begin{figure}[!htbp]
\begin{minipage}[c]{0.45\textwidth}
\includegraphics[scale=.38]{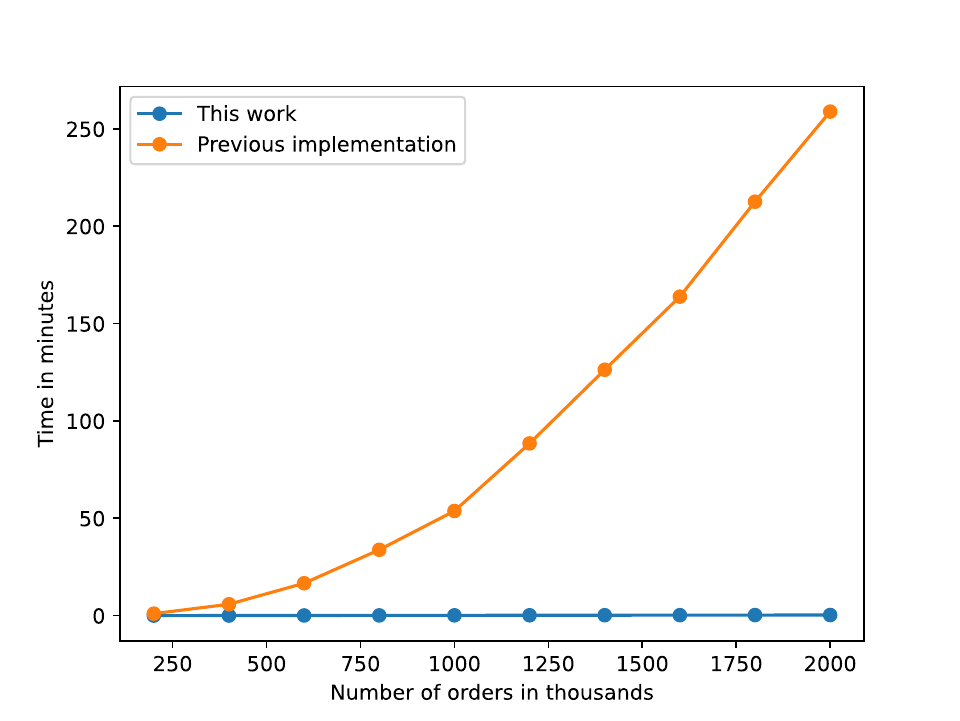}
\end{minipage}\hfill
\begin{minipage}[c]{0.45\textwidth}
\includegraphics[scale=.38]{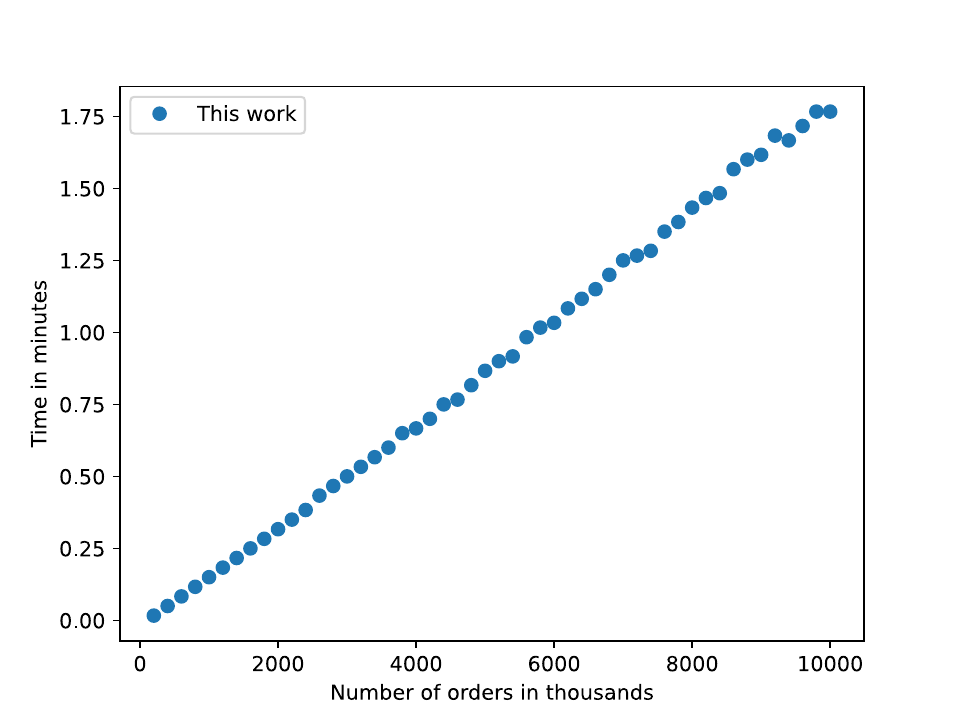}
\end{minipage}\hfill
\end{figure}

The checker obtained in~\cite{GS22} runs reasonably fast on tens of thousands of orders taking only a few seconds and can be useful for products that have a reasonable trade frequency. On the contrary, for products that are traded much more frequently, say receiving ten million orders a day, the checker would take a few days to detect violations, limiting its applicability.
The implementation uses the List data structure resulting in an $O(n^2)$-time implementation for $n$ orders.

In this work, we drastically improve the running time by producing an $O(n\log n)$-time implementation of continuous double auctions using Red-black trees instead of Lists. We formally prove in the Coq proof assistant that the previous verified implementation and our implementation will produce exactly the same output, thus showing adherence to the specifications. 
Furthermore, we prove that any algorithm must take $\Omega(n\log n)$ time showing that the running time of our implementation is asymptotically tight. We run the two implementations on randomly generated data and report the running times. In particular, for ten million orders our implementation takes less than two minutes to run, at least a thousandfold improvement.

Our new implementation, like the previous one, gives rise to an automated checker, which is guaranteed to detect any violation of the specifications from the trade logs if there exists one.

{\bf Organization of the rest of the paper.} In Section 2, we provide a background of continuous double auctions, that we later use to build on.
In Section 3, we state the results that we obtain in this work. In Section 4 we describe the natural algorithm for continuous double auctions with an emphasis on the running time of the previous implementation and the new implementation. In Section 5 we describe our improved implementation. In Section 6 we show  our implementation is efficient and compare the running time of the new implementation with the previous one. Our Coq formalization and experimental demonstrations are available on~\cite{ecda}.
\section{Background}
We now describe the model of an exchange where a fixed product (for example, a particular stock of a company) is traded.
The exchange receives instructions from the traders (buyers and sellers) of the product. Each instruction comprises of a command and an order. Buy, Sell, and Delete are the possible commands.
 An order $w = (i,t,p,q)$ accompanying a Buy or Sell command comprises of a unique identification number $i$, a unique timestamp $t$, a limit price $p$, and a maximum quantity $q$. A Buy $(i,t,p,q)$ instruction from a trader, say John, means that at time $t$ John requests the exchange to buy $q$ units of the product, and he has a budget of $p$ (cents) per unit of the product. This instruction from John is assigned a unique id $i$ by the exchange. Similarly, a Sell $(i,t,p,q)$ instruction from Mary means that at time $t$ she requests the exchange to sell $q$ units of the product for at least $p$ (cents) per unit. Her Sell order is assigned the id $i$ by the exchange. Orders accompanying Delete commands consist of an id and timestamp, but no price or quantity: $(i,t,*,*)$. Buy and Sell orders are traditionally referred to as bids and asks, respectively.
 
 On receiving an instruction, say a bid Buy $w=(i,t,p,q)$ from John, the exchange instantly matches $w$ with existing unmatched or partially matched asks that arrived before and are still in the system (resident asks) and have a limit price of at most $p$. If there is more than one ask that is matchable with $w$, they are prioritized based on price-time priority (first by the competitiveness of their price and then by their timestamps). Transactions between $w$ and these matchable asks are generated of the largest possible quantity, which is at most $q$. These transactions form a matching. If the matching is of total quantity $q$, John's order is completely exhausted and it leaves the system completely. Otherwise, if the size of the matching is $q' < q$, the system keeps the bid $(i,t,p,q-q')$ as a resident bid for future matches. A resident order leaves the system if its quantity gets exhausted by future matches or if it gets deleted by a Delete command. On receiving an ask, the behavior of the exchange is symmetric, where it matches the ask with the resident bids based on their price-time priority.

The exchange algorithm can be thought of as an online algorithm $P$ that maintains a set of resident bids $B$ and a set of resident asks $A$ and gets one instruction at a time. On receiving a new instruction Command $w$, it generates a matching $M$ and updates the sets of resident bids and asks to $\hat B$ and $\hat A$.

$$ (B, A, \text{ Command } w) \stackrel{P}{\mapsto} (\hat B, \hat A, M)$$

Given the above description of the exchange, it was observed in \cite{GS22} that the following three properties must be satisfied by $P$.

\begin{itemize}
    \item Positive bid-ask spread: The most competitive bid in $\hat B$ has a lower limit price than the most competitive ask in $\hat A$, i.e., no transaction is possible among resident orders.
    \item Price-time priority: If Command $w$ is a bid: if $w$ gets traded with an ask $a'$, then each resident ask $a\in A$ which is more competitive than $a'$ must get fully matched in $M$. A symmetric statement holds when Command $w$ is an ask.
    \item Conservation: The system does honest arithmetic and does not modify orders arbitrarily. For example, if a bid $b=(i,t,p,q)$ gets $q'<q$ quantity traded with $w$, then $(i,t,p,q-q')\in \hat B$.
\end{itemize}

These properties are formally represented in Coq as shown below. We do not expect the reader to fully comprehend the following without going through the formal definitions of the various quantities that appear in the previous tool documentation~\cite{addl}. In what follows, $B'$ and $A'$ are obtained from $B$ and $A$ by adding (or removing) $w$ from $B$ and $A$ depending on whether the accompanying command is Buy or Sell (or delete).

\begin{figure}[!htbp]
\includegraphics[scale=.99]{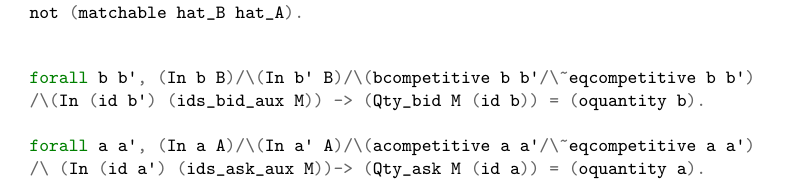}
\end{figure}

\begin{figure}[!htbp]
\includegraphics[scale=.99]{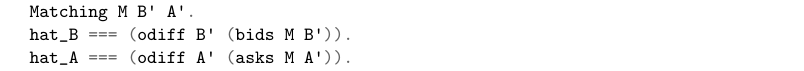}
\end{figure}

The main result in~\cite{GS22} can roughly be summarized as follows. 
\begin{theorem}\label{thm:uniqueness}
Let $P_1$ and $P_2$ be two online algorithms such that each of them satisfies the above three properties. Then, on the same list of instructions as input, at each point in time, $P_1$ and $P_2$ will generate the same matchings.
\end{theorem}

Thus, the above three properties can be used as specifications for continuous double auctions.

A program is also provided in~\cite{GS22}, namely Process\_instruction, and the following theorem is formally proved.

\begin{theorem}\label{thm:consistency}
Process\_instruction satisfies the above three properties.
\end{theorem}

The exchange maintains two logbooks. All the instructions received by the exchange are kept in an order book (sorted by their timestamps) and the corresponding matchings that get generated are maintained in a trade book. As mentioned earlier, a verified program such as Process\_instruction can also be used to build automated checkers that detect violations in trade logs of existing exchanges in an offline mode (like at the end of the day of trading). In fact, the above theorems work in a slightly more general setting where the id of a buy or sell order can be equal to the id of an immediately preceding delete instruction. This allows one to implement Update instructions. Using this slightly general model,~\cite{GS22} implements a checker that can handle more complex orders like Updates, Immediate-or-Cancel orders, market orders, and stop-loss orders by adding a preprocessing step where complex instructions are converted into the three primitive types: Buy, Sell, and Delete. 

Our results hold for this general setting, and all their results still apply in our faster implementation of Process\_instruction.

\section{Our contributions}
We provide a new and efficient implementation of Process\_instruction, which we call eProcess\_instruction, which stores the set of resident bids and asks as red-black trees instead of lists.
We expected changing lists to red-black trees in the previous implementation to be a straightforward task. But unfortunately, it was not so. While the new implementation needed little innovation (like keeping two trees with different keys instead of one, as we will see later), proving the correctness needed both conceptual and technical work. If one carefully observes the specifications for the online process obtained in the previous work~\cite{GS22}, it is formulated in terms of lists. Thus, it is not directly possible to prove the correctness of the implementation that uses trees; one option would be to recast the specifications in terms of trees, but then one would need to prove that the two specifications are equivalent in some sense. Instead, and this is our conceptual innovation, we prove that for any order book, the new and old processes have precisely the same outputs at each point in time, piggybacking on the formal correctness of the previous implementation. Note that we do not prove that our algorithm satisfies the specification directly. Our proof needs to delicately factor in the definition of a `structured' order book, which we explain later.

 On the technical side, we found that the semantic guarantees provided for red-black trees in the standard library implementation of Coq have serious omissions, making the standard library implementation of red-black trees unsuitable for black-box use. Interestingly, we could not find other works using this standard library implementation. We work through the standard library implementation of red-black tees and prove the required guarantees. Note that, through this approach, the running time guarantees of the standard library implementation are thus retained. Others can benefit from our proofs when working with the standard library implementation of red-black trees.

To describe our main results more formally, we first need the following definition.

\begin{definition}[$\cdatree$, $\cdalist$]
Given an order book I, and a natural number $k \leq \mathsf{length}(I)$, let $\cdalist(I,k)$ denote the $k^{\text{th}}$ matching output by Process\_instruction, when it is run on the order book $I$, that is, the matching outputted when it processes the $k^{\text{th}}$ instruction from $I$.
Similarly, $\cdatree(I,k)$ represents the $k^\text{th}$ matching outputted by eProcess\_instruction when run on the order book $I$.
\end{definition}

We are now ready to state our main result.

\begin{theorem}\label{thm:main}
For all order books $I$, and natural numbers $k\leq \mathsf{length}(I)$,\\ $\cdatree(I,k) = \cdalist(I,k).$
\end{theorem}

This theorem appears in our formalization as follows.

\begin{figure}[!htbp]
\includegraphics[scale=.99]{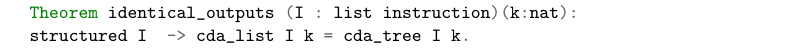}
\end{figure}

The condition that the order book $I$ is structured captures the fact that the timestamps of the orders in $I$ are increasing, and the ids of all bids and asks are distinct, except if it is preceded by a delete instruction, in which case its id could be the same as the id of the preceding delete instruction. As explained earlier, this relaxation allows one to implement more complex instruction types by converting them into the three primitive types, which is useful for certain exchanges.

eProcess\_instruction satisfies the specifications is immediate by combining Theorems~\ref{thm:main} and~\ref{thm:consistency}. Next, we state the time complexity results we obtain.

\begin{theorem}\label{thm:runningTime}
eProcess\_instruction when run on an order book of length $n$ takes $O(n\log n)$ time.
\end{theorem}

\begin{theorem}\label{thm:lowerBound}
Any algorithm that implements continuous double auctions has running time $\Omega(n \log n)$, where $n$ is the length of the order book.
\end{theorem}

Apart from the above results, we add several lemmas that strengthen the guarantees provided in Coq's standard library for red-black trees. This contribution is explained in detail in Section~\ref{sli}.

Our formalization of Theorem~\ref{thm:main} builds on the earlier formalization of~\cite{GS22}. We add about 1500 lines of new Coq code with about 100 new 
lemmas, theorems, and definitions. We use Coq's program extraction feature to obtain an OCaml program for eProcess\_instruction. We use Coq 8.17.1~\cite{coq} for compiling our code. We randomly generate order books of various sizes using a python script, run the extracted verified OCaml programs of eProcess\_instruction and Process\_instruction on it, and report the running times. We include the Coq formalization and the scripts that enable a demonstration of running the extracted programs on the randomly generated data as part of the accompanying materials~\cite{ecda}.
\section{Algorithm for continuous double auctions}

We first describe the algorithm for Process\_instruction as used in~\cite{GS22} and analyze its running time. The algorithm takes as input the set of resident bids B, the set of resident asks A, and an Instruction $\tau$. Depending on whether $\tau$ is a Delete, Buy, or Sell instruction, an appropriate subroutine is called. In the implementation, $B$ and $A$ are kept as sorted lists.

\begin{algorithm}[H]
\caption{Process for continuous double auction}\label{label:processInstruction}
\begin{algorithmic} 
\Function{Process\_instruction}{Bids $B$, Asks $A$, Instruction $\tau$}
\If {$\tau= \mathsf{Del} \ id$} $\mathsf{Del\_order}(B,A,id)$
\EndIf 
\If {$\tau = \mathsf{Buy} \ \beta$}  $\mathsf{Match\_bid}(B,A,\beta)$
\EndIf
 \If {$\tau = \mathsf{Sell} \ \alpha$}
 $\mathsf{Match\_ask}(B,A,\alpha)$
 \EndIf
\EndFunction    
\end{algorithmic}
\end{algorithm}

Next, we present the Match\_ask and Del\_order subroutines as they appear in~\cite{GS22}. Match\_bid is symmetric to Match\_ask and we do not present it explicitly here.

\begin{algorithm}[H]
\caption{Matching an ask}\label{alg:match}
\begin{algorithmic} 
\Function{Match\_ask}{Bids $B$, Asks $A$, order $\alpha$} \Comment{$\alpha$ is an ask.}
\If{$B = \emptyset$} \Return $(B, A\cup\{\alpha\}, \emptyset)$
	\EndIf\\
	\State $\beta \gets \mathsf{Extract\_most\_competitive}(B)$ \Comment{Note: $B \gets B \setminus \{\beta\}$.}\\
	\If {$\pr(\beta) < \pr(\alpha)$}
	     \Return $(B\cup\{\beta\}, A\cup\{\alpha\}, \emptyset)$
	 \EndIf \\
	 
	 \Comment{From now on $\beta$ and $\alpha$ are tradable.}
     \\
	    \If {$\q(\beta) = \q(\alpha)$} \\
	         \hspace{50pt}$m \gets (\id(\beta),\id(\alpha),\q(\alpha))$
	        \State \hspace{17pt}\Return $(B,A,\{m\})$
	       \EndIf
	       \If{$\q(\beta) > \q(\alpha)$}
	             \\ \hspace{50pt}$m \gets (\id(\beta),\id(\alpha),\q(\alpha))$ 
	            \State \hspace{17pt} $B'\gets B\cup\{(\id(\beta), \tsp(\beta), \q(\beta)-\q(\alpha),\pr(\beta))\}$ 
	            \State \hspace{17pt}\Return $(B', A, \{m\})$
	       \EndIf
	       
	    \If {$\q(\beta) < \q(\alpha)$} \\
	    	 \hspace{50pt}$m \gets (\id(\beta), \id(\alpha), \q(\beta))$
	        \State \hspace{17pt}$\alpha'\gets (\id(\alpha),\tsp(\alpha),\q(\alpha)-\q(\beta),\pr(\alpha))$
	        \State \hspace{17pt}$(B', A', M') \gets \mathsf{Match\_ask}(B,A,\alpha')$
	        \State \hspace{17pt} $M \gets M'\cup\{m\}$
	        \State \hspace{17pt}\Return $(B', A', M)$
	     \EndIf

\EndFunction
\end{algorithmic}
\end{algorithm}

Match\_ask takes as input the sets of resident bids and asks $B$ and $A$ (implemented as lists that are kept sorted based on competitiveness), and an ask $\alpha$.
Match\_ask first extracts the most competitive bid $\beta$ from $B$, which happens to be at the topmost element of the list $B$, as $B$ is sorted (which takes $O(1)$ time).
If the limit price of $\beta$ is less than that of $\alpha$, then $\alpha$ is not matchable with any of the orders in $B$. Thus, $\beta$ is inserted back in $B$ (this again takes $O(1)$ time, and $B$ remains sorted) and $\alpha$ must become a resident ask, and is inserted in $A$, which is a sorted list. In the implementation, $\alpha$ is inserted in such a way that the resulting list remains sorted. This takes $O(|A|) = O(n)$ time (assuming there are at most $n$ instructions in total, $|A| + |B| \leq n$) and $\Theta(n)$ time in the worst case (for example, when $\alpha$ needs to be inserted at the middle of the list). This already shows that Process\_instruction's implementation takes at least $\Omega(n)$ time.

Else, $\beta$ is matched with $\alpha$. If the quantity of $\beta$ is at least the quantity of $\alpha$, then $\alpha$ gets completely traded with $\beta$. $\beta$ with its remaining quantity if any is inserted back to $B$. This entire step takes $O(1)$ time.

Finally, if the quantity of $\beta$ is less than $\alpha$, then $\beta$ gets completely matched with $\alpha$. We recursively call match\_ask with the remaining quantity of $\alpha$ with the resident bids and asks $B$ and $A$ (note that $\beta$ is not in $B$ anymore, as it was extracted out at the very beginning). The time taken in this entire step is $O(1)$ plus the time taken by the recursive call. The entire recursion can take at most $O(n)$ time, since the size of $B$ decreases by $1$ for each recursive call.

The Delete instruction simply searches the id in the resident orders and deletes all elements with that id (note that there cannot be more than one such item). This step also takes $O(n)$ time.

\begin{algorithm}[H]
\caption{Deleting an order}\label{alg:delete}
 \begin{algorithmic} 
 \Function{Del\_order}{B, A, id}
 	    \If {$id \in \ids(B)$}
              $B \gets \mathsf{remove}(B, id)$	    
 	    \EndIf
         \If {$id \in \ids(A)$}
              $A \gets \mathsf{remove}(A, id)$	        
 	    \EndIf
 	        \State \Return $(B,A,\emptyset)$
 \EndFunction
 \end{algorithmic}

\end{algorithm}

In total the running time of the list implementation for processing one instruction is $O(n)$. Thus, for processing $n$ instructions, it will take $O(n^2)$ time. The main bottlenecks are inserting an order in a sorted list and deleting an order from the list. The recursive step in Match\_ask also seems to take $O(n)$ time, but one can do a better analysis for multiple instructions to get an improved amortized running time. 

The above online algorithm is used for just processing a single instruction. When we run an online algorithm $P$ repeatedly on an order book $I$, the output matching generated after processing the $k^\text{th}$ instruction is given by the following recursive program of~\cite{GS22}.

\begin{algorithm}[H]
\caption{Iteratively running a process on an order-book}\label{process}
\begin{algorithmic} 
\Function{$\iter$}{Process $P$, Order-book $\mathcal I$, natural number  $k$} \\
\Comment{{promise: $k\leq\mathsf{length}(\mathcal I)$}}
\If {$k=0$} 
\Return $(\emptyset, \emptyset, \emptyset)$
\EndIf 

    \State $(B,A,M) \gets \iter(P, \mathcal I, k-1)$ \\
    \Comment{Note: $B$ and $A$ are the resident bids and asks and $M$ is the matching outputted at time $k-1$.}
    \State $\tau \gets k^{\text{th}} \text{ instruction in } \mathcal I$

    \State \Return $P(B,A,\tau)$
\EndFunction    
\end{algorithmic}
\end{algorithm}

The function $\cdalist$ defined earlier is then obtained by setting $P$ to Process\_instruction.

\subsection{Improved implementation using balanced binary search trees}

We next show how the above-mentioned bottlenecks can be removed by using  balanced binary search trees (BSTs) to store the resident orders, where insertions and deletions take $O(\log n)$ time.

We store the resident bids and asks as BSTs. Note that in Match\_ask the insertion of the ask $\alpha$ in $A$ is done based on competitiveness, whereas in Del\_order the deletion is done using the id. We do not know how to implement insertion  based on one key (competitiveness) and deletion based on another key (id) in $O(\log n)$ time each in the same BST. Consequently, we will keep two BSTs for the same set of resident asks, one for competitiveness and one for id. Similarly, we will have two BSTs for the resident bids. With this trick of keeping two trees for the same set of elements, we can bring down the cost of insertion and the cost of deletion from $O(n)$ to $O(\log n)$.
However, the time to extract the most competitive order (for example $\beta$ from $B$ in the first step of Match\_ask) increases from $O(1)$ to $O(\log n)$. But, this will not hurt the asymptotic running time. 

We always keep the two BSTs for the resident asks (bids) the same. When inserting an order, we insert it in both the trees, which takes $O(\log n)$ time. When deleting an order given by a Delete id command, we first search for the order in the id tree. If we find an order
with that id, we get the order's price and time and search for it in the other tree (whose key is determined by price and time), and then delete that order from both trees, which takes $O(\log n)$ time in total. Similarly, when extracting the most competitive ask, we first extract it from the tree that is ordered by competitiveness. Once the most competitive ask is extracted, we have its id, and then we can easily extract it from the id tree. In total, this takes $O(\log n)$ time.

We now prove that the running time of our
 implementation of the algorithm when it is repeatedly applied to $n$ instructions one after the other is $O(n\log n)$.

\begin{proof}[Proof of Theorem~\ref{thm:runningTime}]
Since there are $n$ instructions in total, the number of resident orders at any point is $|A|+|B|\leq n$. Thus, each insertion, deletion, and extracting the most competitive order takes $O(\log n)$ time.

First notice that there can be at most $n$ Delete instructions, each taking $O(\log n)$ time, needing $O(n\log n)$ time in total.

Now, we will bound the running time for all Match\_ask calls. A similar bound holds for Match\_bid.

Match\_ask calls can be triggered in two ways. Either the call is triggered by a Sell $\alpha$ instruction or it is a recursive call made inside another Match\_ask call. In the latter case, observe that a most-competitive bid $\beta$ gets completely exhausted and leaves the system while executing the calling Match\_ask function. We will say this $\beta$ triggers the recursive call of Match\_ask. Note that a single bid can trigger at most one Match\_ask recursive call.

Time taken to execute a Match\_ask call (not counting the time taken for executing the recursive calls if any) is the time taken for extracting the most-competitive bid $\beta$ and the time taken to insert back the modified $\beta$ in $B$ or the modified ask $\alpha$ in $A$, which in total is $O(\log n)$. As noted above, a Match\_ask is either triggered by a Sell $\alpha$ instruction, or by a bid $\beta$. Since there are at most $n$ instructions and at most $n$ bids, the total time taken by all Match\_ask calls put together is $O(n\log n)$. Note that it might happen that just a single instruction takes $O(n \log n)$ time by itself. But this cannot happen for every instruction.

Thus, Theorem~\ref{thm:runningTime} follows immediately. 
\end{proof}

\section{Implementation with red-black trees}

Our implementation, namely eProcess\_instruction, thus has the following type signature:
$$ (B, B_{id}, A, A_{id}, \text{Command } w) \mapsto (\hat B, \hat B_{id}, \hat A, \hat A_{id}, M),$$
where $B$ and $B_{id}$ are the BSTs for the resident bids with competitiveness and id as the keys, respectively. Similarly, $A$ and $A_{id}$ are the BSTs for the resident asks. The output also contains the four trees for the resident bids and asks after the instruction command $w$ is processed, and the outputted matching $M$, which is maintained as a list.

In our Coq implementation, eProcess\_instruction is defined as follows.

\begin{figure}[!htbp]
\includegraphics[scale=.99]{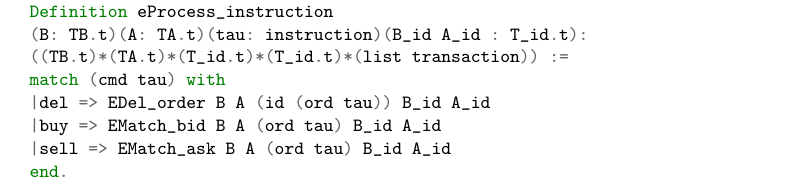}
\end{figure}

Note that since the type signatures of Process\_instruction and eProcess\_instruction are different, we cannot formally show that they are semantically the same. Thus, we show that when we run them on an order book, they have exactly the same outputs at each point in time. For this, we define a version of Iterated for eProcess\_instruction which does exactly the same thing as Iterated to obtain $\cdatree$ and show that $\cdatree$ and $\cdalist$ are semantically the same.

\begin{remark} We make use of the Equations plugin~\cite{equations} for implementing EMatch\_ask and EMatch\_bid smoothly, which was not critical in the previous implementation. Since our new implementation uses red-black trees, and the remove operation on red-black trees in the standard library implementation is not structurally recursive, several additional proof obligations get generated. These obligations are mitigated by using Equations.
\end{remark}

\subsection{Strengthening the standard library guarantees}\label{sli}
For our BSTs, we use the standard library implementation of red-black trees~\cite{redblack}, which implements insertion, extraction, and deletions in $O(\log n)$ time.

However, we could not use the implementation in a black-box manner. To illustrate this consider the following lemma included in the standard library.

\begin{figure}[!htbp]
\includegraphics[scale=.99]{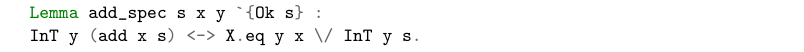}
\end{figure}

Here, $s$ is a red-black tree ($\mathsf{Ok}$ $s$ asserts that), $x$ and $y$ are two elements, $\mathsf{add}$ is a function that inserts $x$ in $s$ and returns the resulting tree.

The above lemma states that an element $y$ is {\it In} the tree after inserting an element $x$ iff $y$ was {\it Equal} to the element $x$ or $y$ was already {\it In} the tree before the insertion of $x$. 

Here, {\it In} and {\it Equal} correspond to $\mathsf{InT}$ and $\mathsf{X.eq}$ in the above lemma, and they do not carry the usual meaning of in and equal.
$\mathsf{InT}$ $a$ $s$ means that the key of $a$ is in the tree $s$, and not that $a$ is in $s$, unless the keys of the elements are the elements themselves.

For our application, the key and the element are not the same; for example, the key is just the id, whereas the element is the entire order.

For using red-black trees, it is easy to imagine situations where one needs the following guarantee. Assume there is an element $x$ whose key is not in the current tree $s$. Then, $x$ will be part of the tree $(\mathsf{add} \ x \ s)$. Furthermore, the other elements of $(\mathsf{add} 
 \ x \ s)$, other than $x$, are precisely the elements of $s$. Such a guarantee is not currently available in the standard library.

One thus requires a stronger statement to prove the correctness of the insert operation $\mathsf{add}$, which was missing in the standard library implementation of red-black trees. Similar issues are there with the specification lemmas for other operations like $\mathsf{remove}$.
Thus, in our formalization, we had to go through the implementation in detail to prove the stronger results needed in our application.
In particular, we show the following where $\mathsf{Intree}$ $y$ $s$ means $y$ is an element of the tree $s$.

\begin{figure}[!htbp]
\includegraphics[scale=.99]{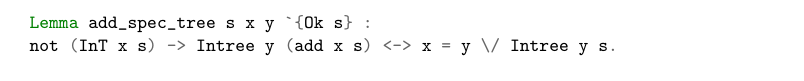}
\end{figure}

The above lemma captures the situation described above. If the key of $x$ is not in the tree $s$, then an element y is part of $(\mathsf{add} \ x \ s)$ if and only if either $x$ and $y$ are the same elements or $y$ is an element in the tree $s$.

Similarly, we strengthen specifications relating to the remove operation and the elements function which outputs a list. They are as follows. Our strengthened lemmas have a suffix of `\_tree'.

\begin{figure}[!htbp]
\includegraphics[scale=.99]{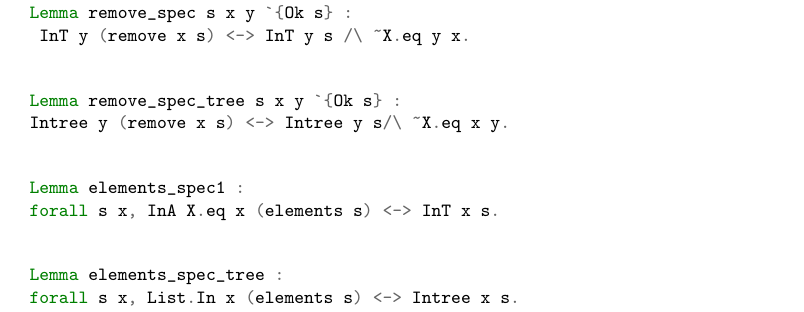}
\end{figure}

Such lemmas derived in our formalization could be useful for future works that use the standard library implementation of red-black trees.

\section{Time complexity of continuous double auctions}

We now show that continuous double auctions take $\Omega(n \log n)$ time for processing $2n$ orders, proving Theorem~\ref{thm:lowerBound}.

\begin{proof}[Proof of Theorem~\ref{thm:lowerBound}]We reduce the task of sorting $n$ positive numbers $\{b_1,\cdots, b_n\}$ in decreasing order to continuous double auctions. Our order book consists of $2n$ orders each with quantity $1$. The first $n$ orders are bids, whereas the last n orders are asks. The $i^\text{th}$ bid has a price $b_i>0$. Each ask has a limit price of $0$.

Observe that each pair of bid and ask is matchable. 
Since the quantity of each order is $1$, there will be exactly $n$ matchings produced each of quantity $1$. The first matching will be produced at the $n+1^\text{th}$ step, where the most-competitive bid will be picked out. At the next step the second most competitive bid is picked out and so on. Hence, the matchings are produced in the sorted order of competitiveness. Since continuous double auctions prioritize matchable orders by price first (and then by time), the trade book generated will be sorted by price. 

We now simply use the folklore sorting lower bound of $\Omega(n \log n)$, to get our result. 
\end{proof}
\subsection{Experimental findings}

We show the running time of Process\_instruction (previous work) and eProcess\_instruction (current work) on randomly generated datasets using a python script which we enclose in the accompanying materials. Our order books generated are of sizes going from $200$ thousand to $10$ million. We stopped running the previous implementation on very large datasets as it took an inordinate amount of time. The figures appearing on page 2 represent the following data graphically.
\vspace{8pt}

\begin{tabular}[H]
  {>{\raggedleft\arraybackslash}p{2 cm}%
  >{\raggedleft\arraybackslash}p{3.5cm}%
  >{\raggedleft\arraybackslash}p{3.5cm}%
  }
\hline
No. of orders & Time (previous) & Time (current) \\
\hline
200,000  &  1 min 01 s  &  2 s  \\
400,000  &  5 min 51 s  &  4 s  \\
600,000  &  16 min 43 s  &  6 s \\
800,000  &  33 min 51 s  &  8 s \\
1,000,000  &  53 min 47 s  &  10 s \\
1,200,000  &  1 hr 29 min  &  12 s \\
1,400,000  &  2 hr 07 min  &  14 s \\
1,600,000  &  2 hr 45 min  &  16 s \\
1,800,000  &  3 hr 33 min  &  18 s \\
2,000,000  &  4 hr 20 min  &  20 s \\ 
\hline
4,000,000  &   &  41 s \\
6,000,000  &   &  1 min 03 s \\
8,000,000  &   &  1 min 27 s \\
10,000,000  &   &  1 min 47 s\\
\hline
\end{tabular}

\vspace{10pt}

As part of the supplementary materials~\cite{ecda}, we provide a demonstration that runs both the old and new implementations on two randomly generated datasets and reports the running times. One needs an OCaml compiler to be able to run this demonstration.

\section{Conclusions}

In this work, we provide an efficient and formally correct implementation for continuous double auctions. This has a drastic impact on the running time as demonstrated by our analysis and leads to fast checkers which are extremely useful for finding errors and monitoring existing exchanges.
\section*{Acknowledgement}
 We wish to thank Mohimenul Kabir of the National University of Singapore for presenting our work at the conference.

\bibliographystyle{plain}
\bibliography{main}

\end{document}